\documentclass[12pt,a4paper]{article}
\usepackage[utf8]{inputenc}
\usepackage{fullpage}
\usepackage{amsmath}
\usepackage{amsfonts}
\usepackage{amssymb}
\usepackage{dsfont}
\usepackage{color}
\usepackage{hyperref}
\usepackage[]{graphicx}
\usepackage{algorithmic}
\usepackage{algorithm}

\floatname{algorithm}{Algorithm}
\usepackage[toc,page]{appendix}

\usepackage{subfig}
\usepackage{natbib}
\author{Matthieu Marbac and Mohammed Sedki}
\title{A  Family of Blockwise One-Factor Distributions for Modelling High-Dimensional Binary Data}

\newcommand{\bX}{\boldsymbol{X}}
\newcommand{\bx}{\boldsymbol{x}}
\newcommand{\tx}{\textbf{x}}
\newcommand{\tu}{\textbf{u}}
\newcommand{\ttt}{\textbf{t}}

\newcommand{\bth}{\boldsymbol{\theta}}
\newcommand{\bdelta}{\boldsymbol{\delta}}
\newcommand{\balpha}{\boldsymbol{\alpha}}
\newcommand{\bvarepsilon}{\boldsymbol{\varepsilon}}
\newcommand{\bTh}{\boldsymbol{\Theta}}
\newcommand{\B}{\textsc{b}}
\newcommand{\bomega}{\boldsymbol{\omega}}
\newcommand{\bomegastar}{\boldsymbol{\omega}^{\star}}
\newcommand{\bOmega}{\boldsymbol{\Omega}}

\newcommand{\BIC}{\mathbf{BIC}}
\newcommand{\argmax}{\mbox{arg\,max}}

\newtheorem{theorem}{Theorem}[section]

\newtheorem{proposition}[theorem]{Proposition}
\newtheorem{corollary}[theorem]{Corollary}

\newenvironment{proof}[1][Proof]{\begin{trivlist}
\item[\hskip \labelsep {\bfseries #1}]}{\end{trivlist}}

\newcommand{\qed}{\nobreak \ifvmode \relax \else
      \ifdim\lastskip<1.5em \hskip-\lastskip
      \hskip1.5em plus0em minus0.5em \fi \nobreak
      \vrule height0.75em width0.5em depth0.25em\fi}

\begin{document}
\maketitle

\begin{abstract}
We introduce a new family of one factor distributions for high-dimensional binary data.
The model provides an explicit probability for each event, thus avoiding the numeric approximations often made by existing methods.
Model interpretation is easy since each variable is described by two continuous parameters (corresponding to its marginal probability and to its strength of dependency with the other variables) and by one binary parameter (defining if the dependencies are positive or negative). 
An extension of this new model is proposed by assuming that the variables are split into independent blocks which follow the new 
one factor distribution.
Parameter estimation is performed by the inference margin procedure where the second step is achieved by an expectation-maximization algorithm.
Model selection is carried out by a deterministic approach which strongly reduces the number of competing models. 
This approach uses a hierarchical ascendant classification of the variables based on the empirical version of Cramer's V for selecting a narrow subset of models.
The consistency of such procedure is shown.
The new model is evaluated on numerical experiments and on a real data set.
The procedure is implemented in the R package MvBinary available on CRAN.
\end{abstract}


\section{Introduction}
Binary data are increasingly emerging in various research fields, particularly in economics, psychometrics or in life sciences \citep{cox1989analysis,collett2002modelling}. 
To carry out statistical inference, it is important to have at hand flexible distributions for such data.
However, there is a shortage of multivariate distributions for binary data \citep{genest2007primer}. Indeed, many approaches have been developed by considering that the binary variables are responses of several explanatory variables \citep{glonek1995multivariate, nikoloulopoulos2008multivariate, genest2013predicting}. However, these models cannot manage data composed with only binary variables.

Since binary variables are easily accessible and poorly discriminative, the binary data sets are often composed of many variables. Thus, the modelling of high-dimensional binary data is an important issue. Moreover, classical models suffer from the \emph{curse of dimensionality} since they involve too many parameters \citep{bellman1957dynamic}. Therefore, specific distributions should be introduced to manage such data.

Many authors have been interested in defining the properties of a multivariate distribution which permit an easy interpretation and inference \citep{nikoloulopoulos2009finite,panagiotelis2012pair}. Thus, \citet{nikoloulopoulos2013copula} lists the five following features that define a distribution with good properties: 
(F1) Wide range of dependence, allowing both positive and negative dependence;
(F2) Flexible dependence, meaning that the number of bivariate marginals is (approximately) equal to the number of dependence parameters;
(F3) Computationally feasible cumulative distribution function for likelihood estimation;
(F4) Closure property under marginalization, meaning that lower-order marginals belong to the same parametric family;
(F5) No joint constraints for the dependence parameters, meaning that the use of covariate functions for the dependence parameters is straightforward.

The modelling by \emph{dependency trees} \citep{chow1968approximating} is a pioneer approach for assessing the distribution of binary variables. A strength of this method is the easy maximization of likelihood function by the Kruskal algorithm \citep{Kru56}, which estimates the tree of minimal length.  Although the tree structure leads to benefits (estimation, visualisation and interpretation), it is limited to simple dependency relations. Moreover, it does not provide parameters for measuring the strength of the dependencies between two variables.

A naïve approach for modelling binary variables is to use a product of Bernoulli distributions. 
However, in spite of the parsimony induced by the independence assumption, this approach leads to severe biases when variables are dependent. 
Thus, a \emph{mixture model} with conditional independence assumption can capture the main dependencies \citep{goodman1974exploratory}.
\citet{celeux1991clustering} propose different parsimonious models to deal with high-dimensional data. 
However, this mixture-based method suffers primarily from a lack of interpretation of dependencies. 
Indeed, there is no parameters for directly reflecting the strength of the dependencies between variables. 

The \emph{quadratic exponential binary distribution} \citep{cox1972analysis} is considered as the binary version of the multivariate Gaussian distribution. However, this model does not retain its exact form under marginalization, but closure under marginalization can be achieved approximately \citep{cox1994note}. This model is not really suitable for high-dimensional data since it implies a quadratic number of parameters.

The modelling of spatial binary data can be achieved by latent Gaussian Markov Random Fields \citep{pettitt2002conditional,weir2000binary} or lattice-based Markov Random Fields like the Ising model \citep{gaetan2010spatial}. 
These approaches can deal with high-dimensional data since they have the Markov properties.
However, their using for non-spatial data is not really doable except when the model at hand is known.
Indeed, this approach requires to define the neighbourhood of each site.
The combinatorial issue of model selection also prevents the use of such approaches when data are non-spatial.

The general approach to model multivariate distributions is to use copulas.
Indeed, \emph{copulas} \citep{nelsen2006,joe1997multivariate} can be used to build a multivariate model by defining, on the one hand, the one-dimensional marginal distributions, and, on the other, the dependency structure. 
Among the copulas, the \emph{Gaussian} and the \emph{Student} ones are very popular since they model the pairwise dependencies. 
However, their likelihood has not a closed form when the variables are discrete. 
It can be approached by numerical integrations which is not doable for high-dimensional data. 
Moreover, they require a quadratic number of parameters which leads to the curse of dimensionality for high-dimensional data. 
Alternatively the \emph{Archimedean copulas} are relevant to reduce the number of parameters since they use a single parameter to model the dependencies between all the variables. Thus, this parameter characterizes a general dependency over the whole variables but it also limits the interpretation. 
For instance, positive and negative dependencies cannot be modelled simultaneously.
Moreover, the evaluation of the likelihood requires the evaluation of an exponential number of terms, so it is not doable for high-dimensional data.
Finally, \emph{vine copulas} \citep{kurowicka2011dependence} are a powerful alternative since they allow the specification of a joint distribution 
on $d$ variables with given margins by specifying
$\binom{d}{2}$  bivariate copula and conditional copula.
Note that the vine copulas generalize and increase the flexibility with respect to the dependencies trees.

The \emph{one factor copulas} \citep{knott1999latent} enable to reduce the number of parameters and thus to deal with high-dimensional data.
This approach assumes that the dependencies between the observed variables are explained by a continuous latent variable.
This approach can be used for modelling continuous data variables \citep{krupskii2013factor}, extreme-value continuous data \citep{mazo2015factor}
or ordinal data \citep{nikoloulopoulos2013factor}.
 
In this work, we introduce a new family of one factor distributions that can be written as a specific one factor copula.
For modelling more complex dependency structures, we extend this family by allowing a partition of the set of observed variables into independent blocks, where each block follows the new one factor distribution. The resulting family respects the five features listed by  \cite{nikoloulopoulos2013copula}.
According to this specific distribution, each variable is described by three parameters: 
a continuous parameter indicating its marginal probability,
 a continuous parameter indicating the strength of the dependency with the rest of variables of the block (through the latent variable) and a discrete parameter indicating if the dependency is positive or negative. 
 
Since the proposed distribution is a specific copula for discrete data, 
parameter inference is achieved by a two step procedure named Inference Function for Margin (IFM, see \citet{joe1997multivariate,joe2005asymptotic}). 
Model selection consists in finding the best partition of the variables into blocks according to the Bayesian Information Criterion (BIC; \citet{schwarz1978bic, neath2012bic}). 
Although this information criterion is defined with the maximum likelihood estimates, an extension has been proposed with the parameter estimates resulting from the IFM \citep{gao2010composite}. 
For high-dimensional data, an exhaustive approach computing the BIC for each possible model is not doable. 
Therefore, we propose a deterministic two step procedure for the model selection. 
First, a small subset of models is extracted from the whole competing models by a deterministic procedure based on a Hierarchical Ascendant Classification (HAC) of the variables by using their empirical Cramer's V. Second, the BIC is computed for the models belonging to this subset and the model maximizing this criteria is returned. We show that this approach is asymptotically consistent.
 Indeed, Metropolis-Hastings algorithm \citep{robert2004monte} is used for detecting the model maximizing the BIC criterion. Alternatively, a Metropolis-Hastings algorithm \citep{robert2004monte} can also be used for 
detecting the model maximizing the BIC criterion. However, we numerically show that the deterministic procedure obtains similar results, in a strongly reduced computing time, as the stochastic one. Therefore, we advise to use the deterministic procedure.

The paper is organised as follows. Section~\ref{sec2} introduces the new family of the specific one factor distributions per independent blocks. Section~\ref{sec3} presents the parameter inference with the IFM procedure. The model selection issue is detailed in Section~\ref{sec4}. Section~\ref{sec5}  numerically compares both model selection procedures.   Section~\ref{sec6} illustrates the approach on a real data set. Section~\ref{sec7} concludes this work. All the mathematical proofs are in appendix. The R package MvBinary implements the proposed method and contains the real data set. It is available on CRAN and the url \url{http://mvbinary.r-forge.r-project.org/} proposes a tutorial for reproducing the application described in Section~\ref{sec6}.

\section{Multivariate distribution of binary variables}
\label{sec2}
\subsection{Blocks of independent variables}
The aim is to model the distribution of the $d$-variate binary vector $\bX=(X_1,\ldots,X_d)$.
Variables are grouped into $\B$ independent blocks for modelling different kinds of dependencies. Thus, the vector $\bomega=(\omega_{1},\ldots,\omega_{d})$ determines the block of each variables since $\omega_j=b$ indicates that $X_j$ is assigned to block $b$ with $1\leq b \leq \B$. Therefore, independence between blocks implies
\begin{equation}
\forall 1\leq j \leq j'\leq d : \omega_j\neq \omega_{j'} \implies \; X_j \perp  X_{j'}.
\end{equation}

Vector $\bomega$ defines a model which is unknown and which has to be infered from the data. Variables affiliated to block $b$ are mutually dependent and are denoted by $\bX_{\{b\}}=(X_j: \omega_j=b)$. Obviously, this approach allows to model dependencies between all the variables (\emph{i.e.} $\B=1$ then $\omega_j=1$ for all $1\leq j \leq d$) or independence between all the variables (\emph{i.e.} $\B=d$ then $\omega_j=j$ for all $1\leq j \leq d$).
The probability mass function (pmf) of the realisation $\bx=(x_1,\ldots,x_d)$ is 
\begin{equation}
\label{bloclik}
p(\bx|\bomega,\bth)=\prod_{b=1}^{\B} p(\bx_{\{b\}} |\bth_b),
\end{equation}
where $\bth=(\bth_b;b=1,\ldots,\B)$ groups the model parameters, where $\bth_b$ groups the parameters of the variables of block $b$. Finally, $p(.|\bth_b)$ is the pmf of variables affiliated to block $b$ and each block is assumed to follow the one-factor distribution described in the following.

\subsection{One-factor distribution per blocks}
\subsubsection{Conditional block distribution}
In block $b$, dependencies between variables are characterised through a random continuous variable $U_b$ which follows a uniform distribution on $[0,1]$. More precisely, variables of block $b$ are independent conditionally on $U_b$. So, the pmf of variables affiliated to block $b$ is 
\begin{equation}
p(\bx_{\{b\}}|u_b,\bth_b)=\prod_{j \in \Omega_b} p(x_j|u_b,\bth_j),  \label{eq:indptcond}
\end{equation}
where $\bth_b=(\bth_j;j\in\bOmega_b)$, $\bth_j$ denotes the parameters related to variable $X_j$ detailed below, and where $\Omega_b=\{j: \omega_j=b\}$ is the set of the indices of the variables of block $b$. 
Therefore, the specific conditional distribution of $\bx_{\{b\}}$ is assumed to be a product of Bernoulli distributions whose parameters are defined according to the value of $u_b$. Indeed, for $j\in\Omega_b$
\begin{equation}
p(x_j|u_b,\bth_j) = p_j^{x_{j}} (1-p_j)^{1-x_{j}}\text{ with } p_j= (1-\varepsilon_{j}) \alpha_{j} + \varepsilon_{j} \mathds{1}_{\{u_b<\alpha_{j}\}}^{\delta_j}\mathds{1}_{\{u_b>1-\alpha_{j}\}}^{1-\delta_j}, \label{eq:conditional}
\end{equation}
where $\bth_j=(\alpha_j,\varepsilon_j,\delta_j)$ groups the parameters related to variable $X_j$ where:
\begin{itemize}
\item the continuous parameter $\alpha_j\in (0,1)$ corresponds to the marginal probability that $X_j=1$ since one can easily verify that for $j\in\bOmega_b$, $\int_0^1 p(X_j=1|u_b,\bth_j) du_b=\alpha_j$,
\item the continuous parameter $\varepsilon_j\in (0,1)$ reflects the dependency strength between $X_j$ and the other variables of block $j$ since the stronger the $\varepsilon_j$, the more correlated are the variables of the block (see Proposition~\ref{prop:dependency}),
\item the binary parameter $\delta_j\in\{0,1\}$ indicates the nature of the dependency, since $\delta_j=1$ if the observed variable is positively dependent with the latent variable and $\delta_j=0$ otherwise. Thus, two variables $X_j$ and $X_{j'}$ affiliated to the same block (\emph{i.e.} $\omega_j=\omega_{j'}$) are positively correlated if $\delta_j=\delta_{j'}$ and they are negatively correlated if $\delta_j=1-\delta_{j'}$.
\end{itemize}
Note that the model identifiability is discussed is the next section.

The parametrization of \eqref{eq:conditional} is convenient for the model interpretation. However, we introduce the following new parametrization which simplifies the likelihood computation. Conditionally on $u_{\omega_j}$, $x_j$ follows a Bernoulli distribution whose the parameters are only determined by a relation between $u_{\omega_j}$ and real $\beta_j=\alpha_j^{\delta_j}(1-\alpha_j)^{1-\delta_j}$ which corresponds to the marginal probability that $X_j=\delta_j$. Indeed, for $u_{\omega_j}\in[0,\beta_j)$, the conditional distribution $X_j|u_{\omega_j},\bth_j$ is a Bernoulli distribution $\mathcal{B}(\lambda_j)$ where $\lambda_j=(1-\varepsilon_j)\alpha_j + \varepsilon_j\delta_j$. Moreover, for $u_{\omega_j}\in[\beta_j,1]$, the conditional distribution $X_j|u_b,\bth_j$ is a Bernoulli $\mathcal{B}(\nu_j)$ where $\nu_j=(1-\varepsilon_j)\alpha_j + \varepsilon_j(1-\delta_j)$. Thus, \eqref{eq:conditional} can be summarized as follows
\begin{equation}
p(x_j|u_b,\bth_j) = \left\{ \begin{array}{rl}
\lambda_j^{x_j}(1-\lambda_j)^{1-x_j} & \text{ if } 0\leq u_b <\beta_j \\
\nu_j^{x_j}(1-\nu_j)^{1-x_j} & \text{ if } \beta_j\leq u_b <1 \\
\end{array} \right. . \label{eq:newconditional}
\end{equation}

\subsubsection{Marginal block distribution}
Obviously, the realizations $u_b$ are not observed, but the distribution of the observed variables $\bX_{\{b\}}$ results from the marginal distribution of the pair $(\bX_{\{b\}},U_b)$. So, the pmf of $\bx_{\{b\}}$  is defined by
\begin{equation}
\label{pdf}
p(\bx_{\{b\}}|\bth_b)=\int_{0}^1 p(\bx_{\{b\}}|u_b,\bth_b) du_b.
\end{equation}
We now describe the properties of the block distribution. All proofs are given in Appendix~\ref{app:propositions}.
For respecting the feature (F3) of \citet{nikoloulopoulos2013copula} and for dealing with high-dimensional data, the block distribution needs to have a closed form. This explicit pmf is detailed in the following proposition.

\begin{proposition}[Explicit distribution] \label{prop:likelihood}
Let $\sigma_b$ be the permutation of $\Omega_b$ such that for $1\leq j < j' \leq d_b$ the following inequality holds $\beta_{(b,j)}\leq\beta_{(b,j')}$, where $\beta_{(b,j)}:=\beta_{\sigma_b(j)}$ and where $d_b=\text{card}(\Omega_b)$ is the number of variables assigned to block $b$.
The integral defined by \eqref{pdf} has the following closed form
\begin{equation} 
p(\bx_{\{b\}} |\bth_b)=\sum_{j=0}^{d_b} (\beta_{(b,j+1)} - \beta_{(b,j)}) f_{b}(j;\bth_b), \label{eq:distribution}
\end{equation}
 where we define $\beta_{(b,0)}=0$ and $\beta_{(b,d_b+1)}=1$. Finally function $f_b(.)$ is defined by
\begin{equation}
f_{b}(j_0;\bth_b)=
\prod_{j=1}^{j_0} \nu_{(b,j)}^{x_{(b,j)}}(1-\nu_{(b,j)})^{1-x_{(b,j)}}
\prod_{j=j_0+1}^{d_b} \lambda_{(b,j)}^{x_{(b,j)}}(1-\lambda_{(b,j)})^{1-x_{(b,j)}} , \label{eq:fbj}
\end{equation}
where  $x_{(b,j)}:=x_{\sigma_b(j)}$ denotes the $j$-th variable (according to permutation $\sigma_b$) assigned to block $b$, $\lambda_{(b,j)}:=\lambda_{\sigma_b(j)}$, $\nu_{(b,j)}:=\nu_{\sigma_b(j)}$ and where $\prod_{j=j_0+1}^{j_0}$ is one.
\end{proposition}

The strength of the proposed model is its easy interpretation. The parameter interpretation is allowed by the property of identifiability now presented.
\begin{proposition}[Model identifiability] \label{prop:ident}
The distribution defined by \eqref{eq:distribution} is identifiable under the following constraints:
\begin{itemize}
\item $\delta_{(b,1)}=1$ if $d_b>2$, $\delta_{(b,1)}=1$; 
\item $\varepsilon_{(b,1)}=\varepsilon_{(b,2)}$ if $d_b=2$; 
\item $\delta_{(b,1)}=1$ and $\varepsilon_{(b,1)}=0$ if $d_b=1$;
\end{itemize}
 where $\delta_{(b,j)}:=\delta_{\sigma_b(j)}$, $\varepsilon_{(b,j)}:=\varepsilon_{\sigma_b(j)}$.
\end{proposition}

The proposed model allows a wide range of dependencies. The following proposition is related the model parameters and Cramer's V. Thus, we can see that the full dependency (respectively anti-dependency) can be modelled by putting $\varepsilon_j=\varepsilon_{j'}$, $\alpha_{j}=\alpha_{j'}$ and $\delta_j=\delta_{j'}$ (respectively $\varepsilon_j=\varepsilon_{j'}$, $\alpha_{j}=1-\alpha_{j'}$ and $\delta_j=1-\delta_{j'}$).
\begin{proposition}[Dependency measures] \label{prop:dependency}
The dependency between two binary variables is often measured with Cramer's V. For the distribution defined by \eqref{eq:distribution}, Cramer's V between $X_j$ and $X_{j'}$ is zero. Moreover, for $j$ and $j'$  and $\beta_{(b,j)}\leq\beta_{(b,j')}$
\begin{equation}
V(X_j,X_{j'})= \varepsilon_{(b,j)}\varepsilon_{(b,j')}
\sqrt{\frac{\beta_{(b,j)}(1-\beta_{(b,j')})}{\beta_{(b,j')}(1-\beta_{(b,j)})}}. \label{eq:cramer}
\end{equation}
\end{proposition}

\section{Parameter inference} \label{sec3}
\subsection{Inference function for Margins}
We observed a sample $\tx=(\bx_1,\ldots, \bx_n)$  assumed to be composed of $n$ independent realizations of the proposed model. The likelihood related to  model $\bomega$ is defined by
\begin{equation}
p(\tx|\bomega,\bth)= \prod_{i=1}^n \prod_{b=1}^{\B} p(\bx_{i\{b\}} |\bth_b).
\end{equation}
The \emph{log-likelihood} function is defined by
\begin{equation}
L(\balpha,\bdelta,\bvarepsilon;\tx,\bomega)=\sum_{i=1}^n  \sum_{b=1}^{\B} \ln p(\bx_{i\{b\}} | \bth_b),
\end{equation}
where $\balpha=(\alpha_j;j=1,\ldots,d)$, $\bdelta=(\delta_j;j=1,\ldots,d)$ and $\bvarepsilon=(\varepsilon_j;j=1,\ldots,d)$.
The proposed distribution is a multivariate copula-based model since each multivariate parametric distribution can be defined as a copula. When the model at hand is a copula with discrete margins, the maximization of the likelihood is quite difficult. Therefore, we use the Inference Function for Margins (IFM) procedure \citep{joe1997multivariate}. This estimation procedure is based on two optimization steps. The first step maximizes the likelihood of univariate margins. The second step maximizes the dependency parameters with the univariate parameters hold fixed from the first step. \citet{joe2005asymptotic} shows the asymptotical efficiency of such a procedure. Thus, the parameters $\hat{\bth}=(\hat{\balpha},\hat{\bdelta},\hat{\bvarepsilon})$ are estimated by the two following steps: \\
\textbf{Margin step:} for $j\in\{1,\ldots,d\}$
$$ \hat{\alpha}_j=\frac{1}{n} \sum_{i=1}^n x_{ij},$$
\textbf{Dependency step:}
$$ (\hat{\bdelta},\hat{\bvarepsilon})=\arg\max_{(\bdelta,\bvarepsilon)} L(\hat{\balpha},\bdelta,\bvarepsilon;\tx,\bomega).$$

The margin step is easily performed, but the search of $(\hat{\bdelta},\hat{\bvarepsilon})$ at the dependency step requires solving equations having no analytical solution (except when $d_b=2$). This step is also achieved by using the latent structure of the data when $d_b>2$ (details are given in Section~\ref{sec:EM}). When $d_b=2$, for $j\in\Omega_b$:
\begin{equation}
\hat{\delta}_{(b,2)}=\left\{ \begin{array}{rl}
1 & \text{ if } n_{11}\geq \hat{\alpha}_{(b,1)}\hat{\alpha}_{(b,2)}\\
0 & \text{ if } n_{11}> \hat{\alpha}_{(b,1)}\hat{\alpha}_{(b,2)}\\
\end{array}\right. \text{ and } 
\hat{\varepsilon}_{(b,1)} = \hat{\varepsilon}_{(b,2)}=
\sqrt{
 \frac{|n_{11} - \hat{\alpha}_{(b,1)}\hat{\alpha}_{(b,2)}|}{\hat{\beta}_{(b,1)}(1-\hat{\beta}_{(b,2)})}},
\end{equation}
where $n_{11} = \frac{1}{n} \sum_{i=1}^n x_{ij_1}x_{ij_2}$ with $j_1\in\bOmega_{b}$,  $j_2\in\bOmega_{b}$ and $j_1\neq j_2$.

\subsection{An EM algorithm for the dependency step} \label{sec:EM}
Since the blocks of the one-factor distributions imply latent variables, it is natural to perform the dependency steps of the IFM procedure with an Expectation-Maximization (EM) algorithm \citep{dempster1977maximum, mclachlan2008em} when $d_b>2$. The \emph{complete-data likelihood} is defined by
\begin{equation}
L(\bth;\tx,\tu,\bomega)= \sum_{j=1}^n L(\bth_j;\tx,\tu,\bomega)
\end{equation}
where 
\begin{align}
L(\alpha_j,\delta_j,\varepsilon_j;\tx,\tu,\bomega) =
&\sum_{i=1}^n z_{ij} \left(
x_{ij} \ln \lambda_j + 
(1-x_{ij}) \ln (1-\lambda_j) \right) \\ \nonumber
& + (1-z_{ij}) \left(
x_{ij} \ln \nu_j + 
(1-x_{ij}) \ln( 1-\nu_j)
\right),
\end{align}
where $z_{ij}=1$ if $0\leq u_{i\omega_j}<\beta_j$ and  $z_{ij}=0$ if $\beta_j\leq u_{i\omega_j}\leq 1$. 
The EM algorithm is an iterative algorithm which alternates between two steps: the computation of conditional expectation of the complete-data log-likelihood (E step) and its maximization (M step) on $(\bdelta,\bvarepsilon)$. Note that the estimate $\hat{\balpha}$ is not modified by the algorithm. Its iteration $[r]$ is written as:\\
\textbf{E step:} Computation of the complete-data log-likelihood, for $j\in\{1,\ldots,d\}$
\begin{equation}
t_{ij}(\bth^{[r]})=\mathbb{E}[Z_{ij}|\bx_i,\bomega,\bth_j^{[r]}] = \frac{\lambda_j^{[r]} \beta_j^{[r]}}{\lambda_j^{[r]} \beta_j^{[r]} + \nu_j^{[r]} (1-\beta_j^{[r]})}.
\end{equation}
\textbf{M step:} Maximization over $(\delta_j,\varepsilon_j)$, for $j\in\{1,\ldots,d\}$
\begin{align}
\delta_j^{[r+1]} & = \mathds{1}_{\{ \max_{\varepsilon_j} L(\hat{\alpha}_j,\delta_j=1,\varepsilon_j;\tx,\ttt^{[r]},\bomega) >  \max_{\varepsilon_j} L(\hat{\alpha}_j,\delta_j,\varepsilon_j;\tx,\ttt^{[r]},\bomega)\}},\\
\varepsilon_j^{[r+1]} &= \argmax_{\varepsilon_j}  L(\hat{\alpha}_j,\delta_j^{[r+1]},\varepsilon_j;\tx,\ttt^{[r]},\bomega),
\end{align}
where $\bth^{[r]}_j=(\hat{\alpha}_j,\delta_j^{[r]},\varepsilon_j^{[r]})$, $\lambda_j^{[r]}=(1-\varepsilon_j^{[r]})\hat{\alpha}_j + \varepsilon_j^{[r] \delta_j^{[r]}}$, $\nu_j^{[r]}=(1-\varepsilon_j^{[r]})\hat{\alpha}_j + \varepsilon_j^{[r] 1-\delta_j^{[r]}}$. Thus, the M step involves the search of the maximum over $\varepsilon_j\in]0,1[$ of $L(\hat{\alpha}_j,\delta_j,\varepsilon_j;\tx,\ttt^{[r]},\bomega)$. This maximization is easily performed since it only leads to solve a quadratic equation as shown by Appendix~\ref{app:EM}.

\section{Model selection} \label{sec4}
\subsection{Information criterion}
Model selection is obviously necessary when we are faced with  model-based statistical inference. 
When the model pmf is given by~\eqref{bloclik}, selecting a model means identifying the repartition of the variables into independent blocks. 
The challenge also consists of finding the best model according to the data among a set of competing models. 
In a Bayesian framework, the best model is defined by the model having the highest posterior probability. 
By assuming that uniformity holds for the prior distribution of $\bomega$, the best model also maximizes the integrated likelihood $p(\tx|\bomega)$ where
\begin{equation}
p(\bomega|\tx) \propto p(\tx|\bomega) \text{ with } p(\tx|\bomega)=\int_{\bTh}  p(\tx|\bomega, \bth) p(\bth|\bomega)d\bth,
\end{equation}
and $p(\bth|\bomega)$ corresponds to the prior distribution of the parameters of model $\bomega$. However, this integral has not a closed form. In thus case, the BIC \citep{schwarz1978bic} is used for approaching the logarithm of the integrated likelihood by using a Laplace approximation. It is defined by
\begin{equation}
\BIC(\bomega)=L(\hat{\bth}_{\bomega};\tx,\bomega) - \frac{\nu_{\bomega}}{2} \ln(n), \label{eq:bic}
\end{equation}
where $\nu_{\bomega}$ corresponds to the number of continuous parameters involved in model $\bomega$ and where $\hat{\bth}_{\bomega}$ is the MLE of model $\bomega$. As shown by~\cite{gao2010composite}, the MLE can be replaced in \eqref{eq:bic} by the estimate provided by the IFM procedure. Thus, we want to obtain model $\bomegastar$ which maximizes the BIC criterion among all the competing models.

The number of competing models is too huge for applying an exhaustive approach. Therefore, Section~\ref{sec:HAC} presents a deterministic procedure for model selection. This procedure applies a filter among the competing models and only selects $d$ models. Then, the BIC criterion is computed for each of the  selected models. We show that this procedure returns the correct model $\bomegastar$ asymptotically with probability one. Moreover, Section~\ref{sec:MH} presents a stochastic algorithm which finds $\bomegastar$. Section~\ref{sec5} shows that both procedures have the same behaviour for detecting the true model, but that the deterministic procedure is drastically faster than the stochastic procedure. Both procedures are implemented in the \textsf{R} package MvBinary, but we advise to only use the deterministic procedure for computing reasons.

\subsection{Deterministic approach for model selection} \label{sec:HAC}
This deterministic procedure has two steps. First, the \emph{reduction step} reduces the number of competing models to only $d$ competing models. Second, the \emph{comparison step} computes the BIC criterion for each of the $d$ resulting models and the model maximizing the BIC criterion is returned.

The reduction step decreases the number of competing models by using the empirical dependencies between the variables. Indeed, it performs the Hierarchical Ascendant Classification (HAC) of the variables based on the empirical Cramer's V. This procedure proposes $d$ partitions corresponding to the $d$ competing models on which the BIC criterion will be computed. Each model proposed by the HAC is relevant since it models the strongest empirical dependencies. Moreover, the HAC provides embedded partitions of variables  
and then reduces the calls to the EM algorithm.

The deterministic procedure based on HAC performs the model selection with the two following steps:\\
\textbf{Reduction step} performs the HAC based on the empiric Cramer's V to defined the $d$ partitions of the variables.\\
\textbf{Comparison step}  computes  $\BIC\big( \bomega^{[k]}\big)$ for $k=1,\ldots,d$, where $\bomega^{[k]}$ is such that each block $b$ is composed by the variables affiliated to class $b$ by the partition into $k$ classes of the HAC.\\
The procedure returns $\displaystyle{\argmax_{k =1,\ldots, d} \BIC\big(\bomega^{[k]}\big)}$.

\begin{proposition}[Consistency of the HAC-based procedure] \label{prop:consistency}
The HAC-based procedure is asymptotically consistent (\emph{i.e.} it returns the true model with probability one when $n$ grows to infinity).
\end{proposition}
Proof is given in Appendix~\ref{app:consistency}.


\subsection{Stochastic approach for model selection} \label{sec:MH}
Model $\bomegastar$ can be determined through a Metropolis-Hastings algorithm \citep{robert2004monte}.
This algorithm performs a random walk over the competing models and its unique invariant distribution is proportional to $\exp{\big(\BIC(\bomega)\big)}$.
Therefore, $\bomegastar$ is the mode of its stationary distribution. It is also sampled with probability one by the algorithm when the number of iterations  $R$ grows to infinity.

At iteration $[r]$, the algorithm samples a model candidate $\tilde{\bomega}$ from the distribution $q(.|\bomega^{[r]})$ where $\bomega^{[r]}$ corresponds to the current model. More precisely, candidate
$\tilde{\bomega}$ is equal to the current model $\bomega^{[r]}$ except for
variable $j^{[r]}$ randomly sampled which is affiliated into block $b^{[r]}$ randomly sampled in $\{1,\ldots,\max(\bomega^{[r]})+1\}$.
 This candidate  is accepted with a probability equal to
\begin{equation}
\rho^{[r]}=\frac{\exp\big(\BIC(\tilde{\bomega})\big) q(\bomega^{[r]}|\tilde{\bomega})}{\exp\big(\BIC(\bomega^{[r]})\big) q(\tilde{\bomega}|\bomega^{[r]})}.
\end{equation}
 This algorithm performs $R$
iterations and returns the model maximizing the BIC criterion. In practice,
there may be almost absorbing states, so different
initialisations of this algorithm ensure to visit $\bomegastar$.
Thus, starting from $\bomega^{[0]}$, uniformly sampled among the competing models, the algorithm performs $R$ iterations and returns $\displaystyle{\argmax_{r=1,\ldots,R} \BIC(\bomega^{[r]})}$. Its iteration $[r]$ performs the two following steps:\\
 \textbf{Candidate step:} $\tilde{\bomega}$ is sampled from $q(.|\bomega^{[r]})$.\\
\textbf{Acceptance/reject step:} defined $\bomega^{[r]}$ with
   	\begin{equation*}
 	\bomega^{[r]} =\left\{ \begin{array}{rl}
 	\tilde{\bomega} & \text{with probability } \rho^{[r]} \\
 	\bomega^{[r-1]} & \text{otherwise}
 	\end{array}\right. .
 	\end{equation*}
  

\section{Numerical experiments}
\label{sec5}

\subsection{Suitability of the HAC-based procedure}
This simulation shows the relevance of competing models provided by the reduction step of the HAC-based procedure. Data are simulated from the proposed model with the following parameters
\begin{equation}
\label{eq:paramsimu}
d=10,\;
\delta_j=1, \;
\alpha_j=0.4
\text{ and } 
\omega_j=\left\{
\begin{array}{rl}
1 & \text{if } 1\leq j\leq 5 \\
2 & \text{if } 6\leq j \leq 10
\end{array}
\right. .
\end{equation}
For different sizes of sample $n$ and strengths of dependencies $\varepsilon_j$, we check if the true model belongs to the list of models returned by reduction step of the HAC-based procedure. Table~\ref{tab:simulation1} shows the results obtained on 50 samples for different values of $(n,\varepsilon_j)$.

\begin{table}[ht!]
\begin{center}
\begin{tabular}{cccccc}
$n|\varepsilon_j$ & 0.2 & 0.3 & 0.4 & 0.5 & 0.6 \\ 
\hline 50 & 0 & 1 & 4 & 28 & 37 \\ 
100 & 1 & 1 & 20 & 41 & 49 \\ 
200 & 0 & 9 & 40 & 48 & 20 \\ 
400 & 2 & 29 & 47 & 50 & 50 \\ 
800 & 3 & 44 & 50 & 50 & 50 \\ 
1600 & 21 & 49 & 50 & 50 & 50 \\ 
3200 & 40 & 50 & 50 & 50 & 50 \\ 
\hline 
\end{tabular} 
\end{center}
\caption{Number of times where the true model belongs to the list of models returned by the reduction step of the HAC-based procedure on 50 samples.} \label{tab:simulation1}
\end{table}
Thus, whatever the strength of dependencies, the procedure asymptotically proposes the true model. Obviously, for a fixed sample size, results are better when the dependencies are strong since the number of times where the true  model belongs to the list of models is increasing with the dependency strength. 

\subsection{Comparison of model selection procedures}
Both procedures of model selection are compared on data sampled from the proposed model with the parameters defined in \eqref{eq:paramsimu}. To compare the quality of the estimates returned by both procedures, we use the Kullback-Leibler divergence. As shown by Table~\ref{tab:simulation2a}, both procedures are consistent since the Kullback-Leibler divergence asymptotically vanishes. Moreover, the estimates have the same quality (equal value of the Kullback-Leibler divergence). However, the HAC-based procedure is considerably faster than the Metropolis-Hastings procedure as shown by  Table~\ref{tab:simulation2b}. So, we recommend to use the HAC-based procedure to perform the model selection in high dimension.

\begin{table}[ht!]
\begin{center}
\begin{tabular}{ccccccccccc}
$n|\varepsilon_j$ & \multicolumn{2}{c}{0.2} & \multicolumn{2}{c}{0.3} & \multicolumn{2}{c}{0.4} & \multicolumn{2}{c}{0.5}\\ 
& HAC & MH & HAC & MH & HAC & MH & HAC & MH \\
\hline 50 & 0.15 & 0.16 & 0.24 & 0.27 & 0.36 & 0.36 & 0.39 & 0.46 \\
100  & 0.08 & 0.09 & 0.13 & 0.14 & 0.22 & 0.21 & 0.12 & 0.17 \\
200  & 0.04 & 0.04 & 0.10 & 0.10 & 0.08 & 0.11 & 0.05 & 0.06 \\
400  & 0.03 & 0.03 & 0.06 & 0.06 & 0.03 & 0.03 & 0.03 & 0.03 \\
\hline 
\end{tabular} 
\end{center}
\caption{Kullback-Leibler divergence obtained with the estimates provided by both procedure of model selection.} \label{tab:simulation2a}
\end{table}

\begin{table}[ht!]
\begin{center}
\begin{tabular}{ccccccccccc}
$n|\varepsilon_j$ & \multicolumn{2}{c}{0.2} & \multicolumn{2}{c}{0.3} & \multicolumn{2}{c}{0.4} & \multicolumn{2}{c}{0.5}\\ 
& HAC & MH & HAC & MH & HAC & MH & HAC & MH \\
\hline 50 & 11 & 217 & 10 & 250 & 9 & 278 & 8 & 381 \\
100  & 12 & 241 & 11 & 250 & 10 & 354 & 8 & 633 \\
200  & 14 & 276 & 13 & 308 & 11 & 662 & 9 & 912 \\
400  & 16 & 296 & 15 & 509 & 12 & 1218 & 9 & 933 \\
\hline 
\end{tabular} 
\end{center}
\caption{Computing time in seconds required by the two procedures of model selection.} \label{tab:simulation2b}
\end{table}

\subsection{Model selection for high-dimensional data}
This section shows the behaviour of the HAC-based procedure in high dimension. Data are generated from a model with blocks of five dependent variables $(d_b=1)$, with equal marginal probabilities ($\alpha_j=0.4$) and equal dependency strength ($\varepsilon_j=0.4$ and $\delta_j=1$). For different sizes of sample and numbers of variables, 50 samples are generated. 

Table~\ref{tab:simulation3} shows the relevance of the deterministic procedure by using the Adjusted Rand Index (ARI) to compare the true partition of the variables into blocks and its estimated.  Indeed, whatever the number of variables, the procedure provides the true model with probability one when $n$ grows to infinity. However, for small samples the procedure can provide a model slightly different to the true model.

\begin{table}[ht!]
\begin{center}
\begin{tabular}{cccccc}
$n|d$ & 10 & 20 & 50 & 100 & 200 \\ 
\hline 50 &  \textbf{0.11} ( 0.18 ) & \textbf{0.07} ( 0.03 ) & \textbf{0.10} ( 0.06 ) & \textbf{0.10} ( 0.04 ) & \textbf{0.06} ( 0.02 ) \\ 
100 & \textbf{0.35} ( 0.33 ) & \textbf{0.35} ( 0.24 ) & \textbf{0.24} ( 0.14 ) & \textbf{0.22} ( 0.07 ) & \textbf{0.15} ( 0.03 ) \\ 
200 & \textbf{0.85} ( 0.27 ) & \textbf{0.78} ( 0.20 ) & \textbf{0.67} ( 0.11 ) & \textbf{0.56} ( 0.07 ) & \textbf{0.43} ( 0.05 ) \\ 
400 & \textbf{0.97} ( 0.09 ) & \textbf{0.95} ( 0.07 ) & \textbf{0.95} ( 0.05 ) & \textbf{0.91} ( 0.05 ) & \textbf{0.86} ( 0.05 ) \\ 
800 & \textbf{1.00} ( 0.00 ) & \textbf{1.00} ( 0.01 ) & \textbf{1.00} ( 0.01 ) & \textbf{1.00} ( 0.01 ) & \textbf{1.00} ( 0.01 ) \\ 
\hline 
\end{tabular} 
\end{center}
\caption{Mean (in bold) and standard deviation (in parenthesis) of the ARI between $\bomega^0$ and $\bomegastar$.} \label{tab:simulation3}
\end{table}

\section{Application to plant distribution in the USA}\label{sec6}
\subsubsection*{Dataset}
Data has been extracted from the USA plants database, July 29, 2015. It describes $35 583$ plants by indicating if they occur in 69 states (USA, Canada, Puerto Rico, Virgin Islands, Greenland and St Pierre and Miquelon). 
By modelling the data distribution, the flora variety of each states could be characterized.
Moreover, one can expect bring out geographic dependencies between the variables. 
The data are available in the \textsf{R} package MvBinary which implements the proposed method.

\subsubsection*{Experiment conditions}
The model selection is achieved by the deterministic algorithm (see Section~\ref{sec:HAC}) where the Ward criterion is used for the HAC. 
The EM algorithm is randomly initialized 40 times and it is stopped when two successive iterations increase the log-likelihood less than 0.01.

\subsubsection*{Model coherence}
Figure~\ref{fig:Cramer} shows the relevance of the dependencies detected by the estimated model. 
Indeed, Figure~\ref{fig:qqplot} shows the correspondence between Cramer's V computed with the model parameter and the empirical Cramer's V, for each pair of variables claimed to be dependent by the estimated model. 
Moreover, Figure~\ref{fig:boxplot} shows that the estimated model well represents the main dependencies.

\begin{figure}[ht!]
  \centering
  \subfloat[Empiric Cramer's V and Cramer's V estimated by the model.]{\label{fig:qqplot}\includegraphics[scale=0.45]{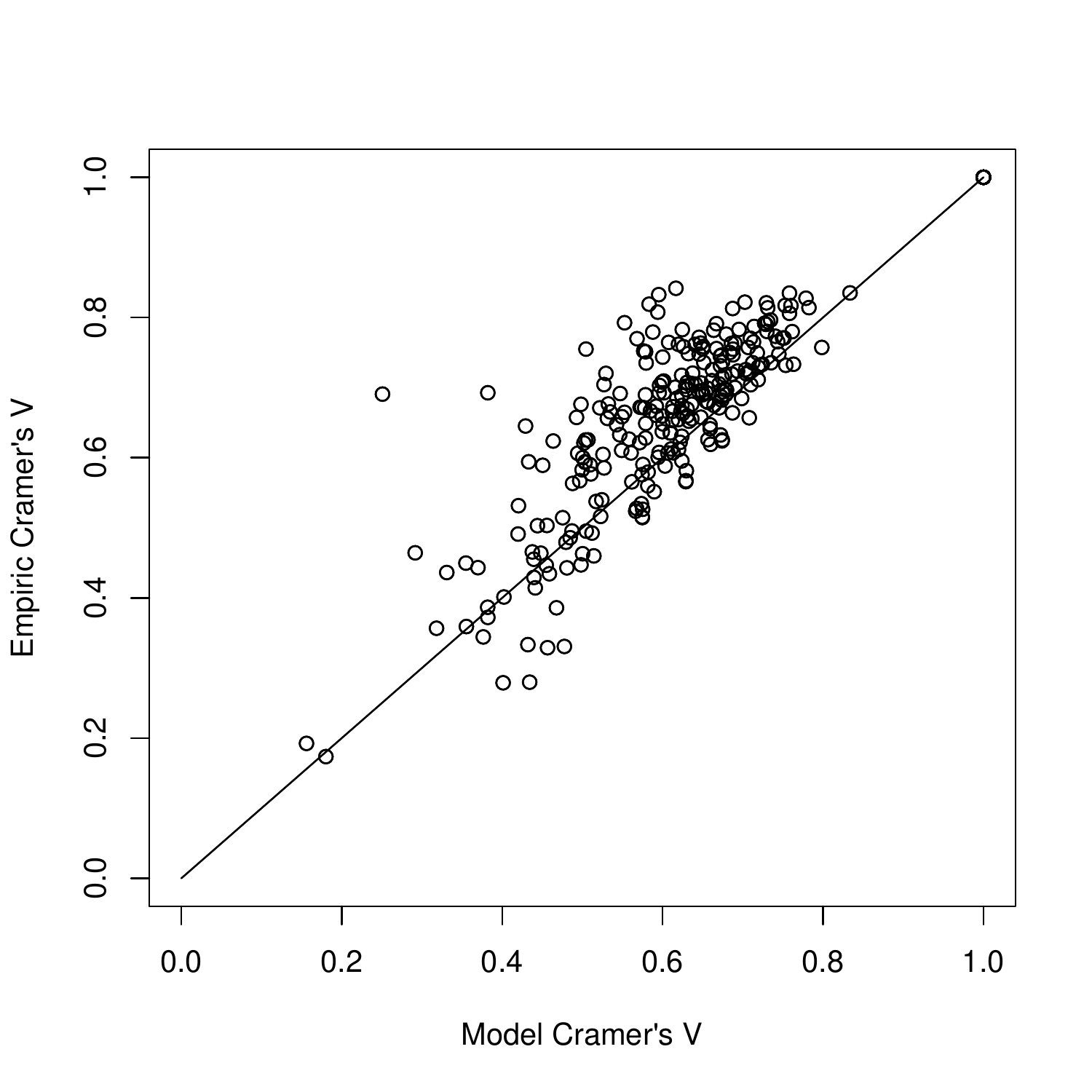}}
  \hspace{5pt}
  \subfloat[Boxplot of the Empiric Cramer's V for the modelled and not modelled dependencies]{\label{fig:boxplot}\includegraphics[scale=0.45]{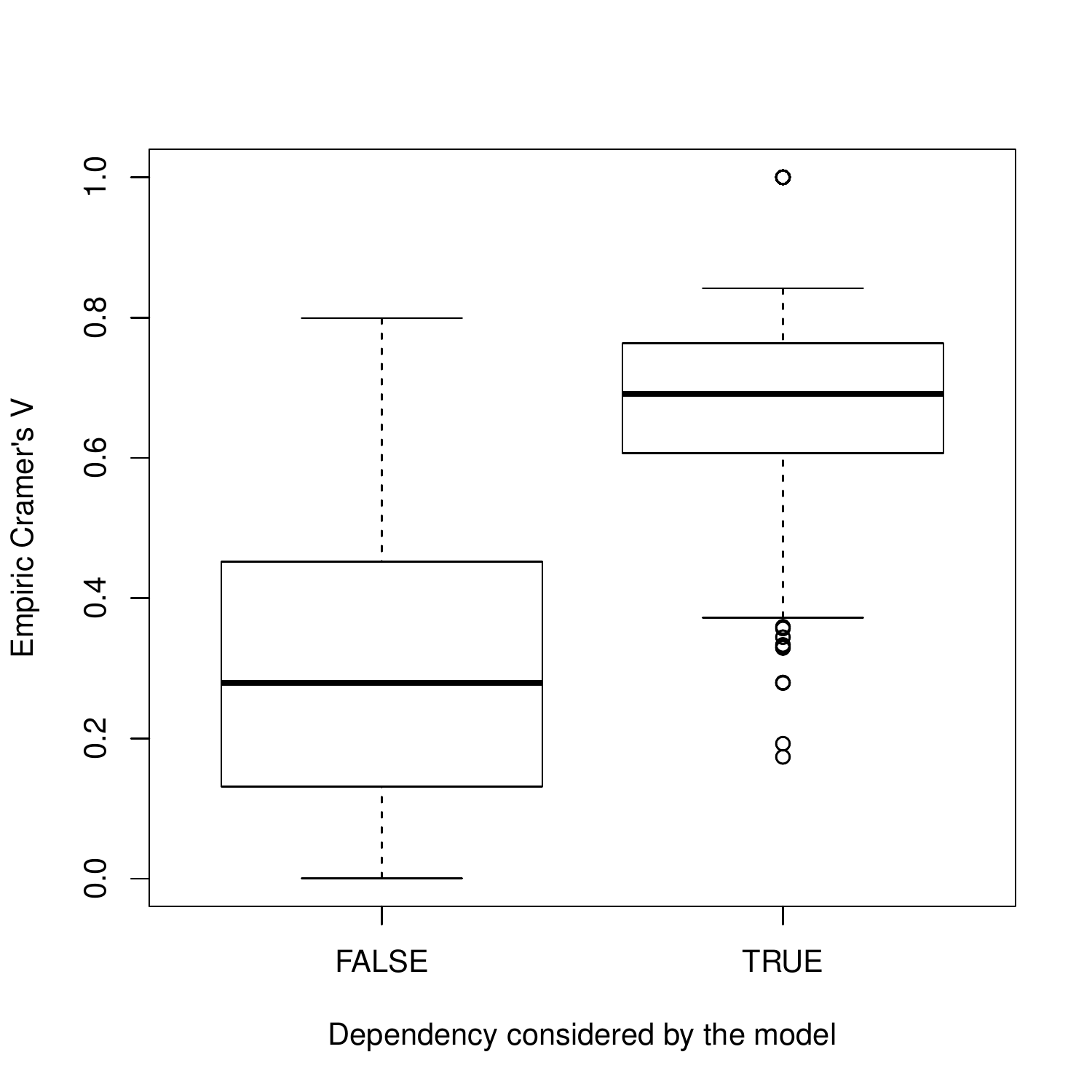}}
  \hspace{5pt}
  \caption{Visualisation of the dependencies taken into account by the model.}
  \label{fig:Cramer}
\end{figure}

The estimated model is composed of 10 blocks of dependent variables. Figure~\ref{fig:map} shows that this block repartition has a geographic meaning.
\begin{figure}[ht!]
\centering \includegraphics[scale=0.4]{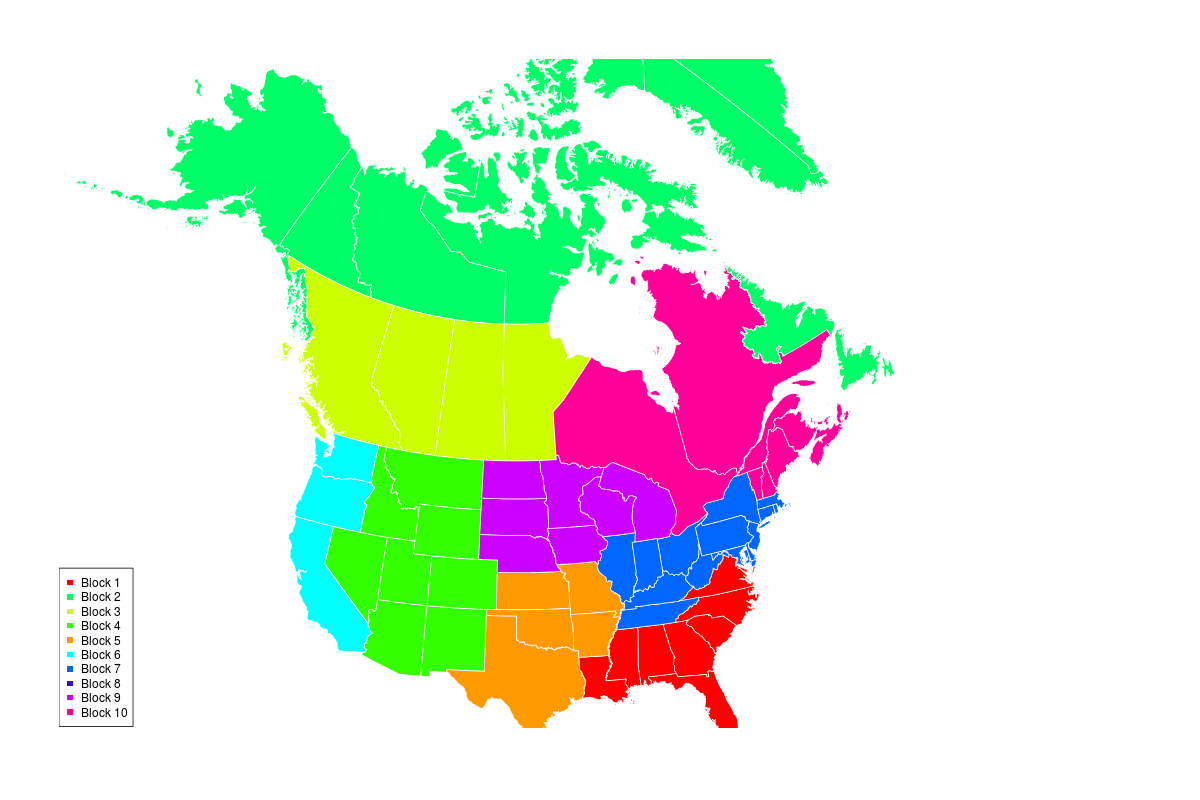}
\caption{Geographic coherence of the blocks of states (color indicates the block assignment)}
\label{fig:map}
\end{figure}

\subsubsection*{Model interpretation}
Parameters permit an easy interpretation of the whole distribution. The mean per block of the values of $\hat{\alpha}_j$ and $\hat{\varepsilon}_j$ are summarized by Figure~\ref{fig:parameters}. Note that the dependencies detected by the model are all positive since for $j=1,\ldots,d$, $\hat{\delta}_j=1$. 

\begin{figure}[ht!]
\centering \includegraphics[scale=0.6]{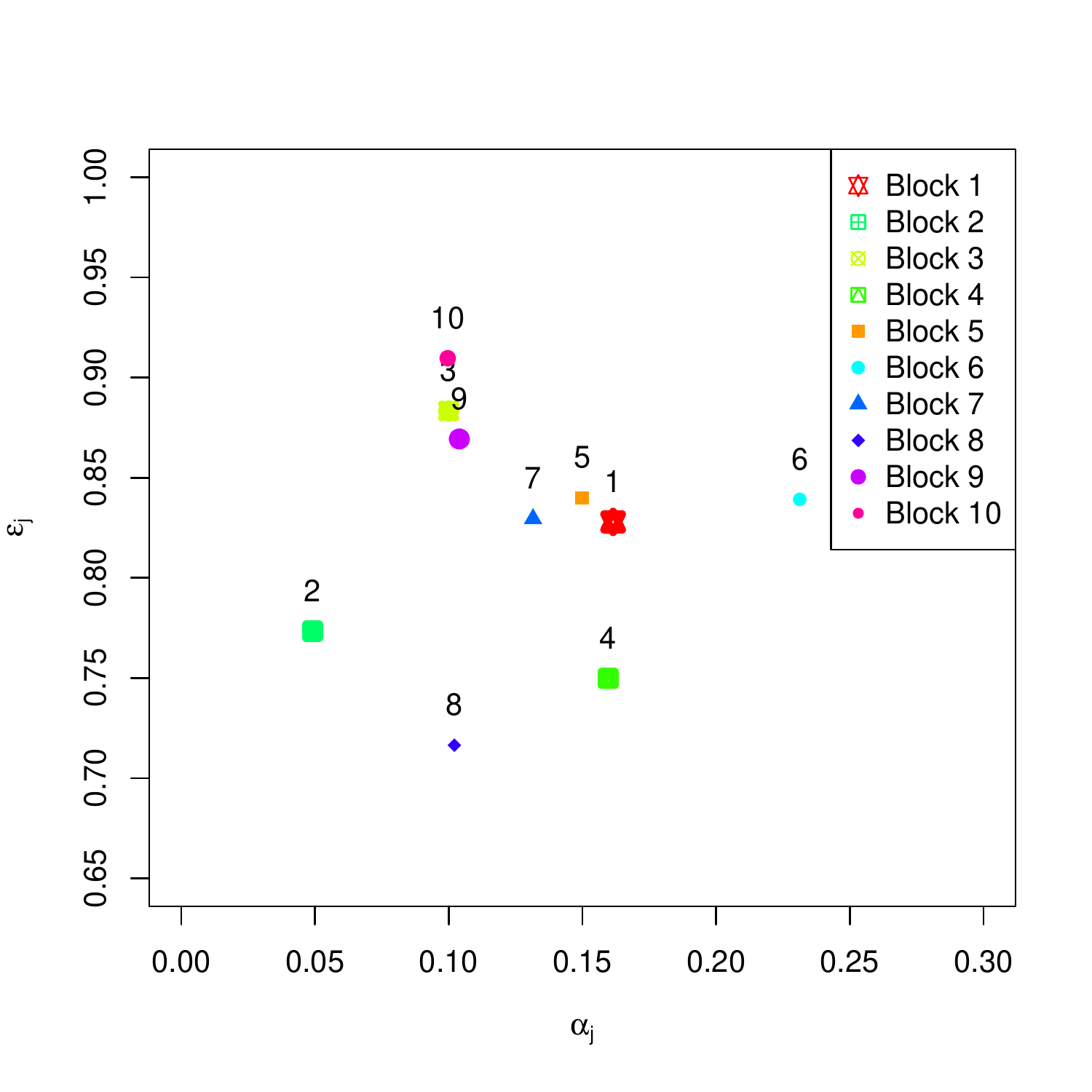}
\caption{Summary of the parameters by blocks}
\label{fig:parameters}
\end{figure}

Each block is composed of highly dependent variables (high values of parameters $\hat{\varepsilon}_j$ and $\hat{\delta}_j=1$). 
Therefore, the knowledge of one variables of a block provides strong information about the other variables affiliated into this block. 
For instance, the most dependent block is Block 10 (composed by Prince Edward Island, Nova Scotia, New Brunswick, New Hampshire, Vermont, Maine, Qu\'ebec and Ontario).
Thus, a plant occurs in Ontario with probability $\hat{\alpha}_{Ontario}=0.14$ while it occurs with a probability $0.83$ if this plant occurs in Qu\'ebec. 
The least dependent block is composed of tropical states (Virgin Islands, Puerto Rico and Hawaii). 
These weaker dependencies can be explained by large geographic distance. 
Finally, parameters $\alpha_j$ allow to described the region by their amount of plants. 
Cold regions (Blocks 2, 3 and 10) obtains small values of $\hat{\alpha}_j$ while the "sun-belt" obtains large values of this parameter.

\section{Conclusion} \label{sec7}
In this paper, we have introduced a new family of distributions for large binary datasets.
This family implies that the variables are grouped into independent blocks and that each block follows a specific one factor distribution.
This new family has many good properties.
Indeed, it verifies the five features required by \citet{nikoloulopoulos2013copula} for a ``good'' distribution.
Moreover, it permits an easy interpretation of the whole distribution.
The variable repartition puts the light on the main dependencies.
Moreover, each variable is summarized by its marginal probability (parameter $\alpha_j$) and by its strength (parameter $\varepsilon_j$) and its kind (parameter $\delta_j$) of dependency with the other block variables.
Finally, this model is suitable for modelling large binary data since its number of parameters is linear in $d$.

We have proposed to circumvent the combinatorial problem of model selection with a deterministic procedure which reduces the number of competing models by using the empirical dependencies.
Although this procedure does not ensure the maximization of the BIC, its consistency has been demonstrated.
Numerical experiments have shown that this approach provides estimates having the same quality as a stochastic (and optimal) procedure, but it strongly reduces the computing time.
The \textsf{R} package MvBinary implements both procedures of inference and contains the data set used in the application.

Many extension of this work can be envisaged.
Indeed, parsimony extensions could be introduced by imposing equality constraints between the block parameters (\emph{e.g} $\forall j\in \Omega_b$, $\varepsilon_j=c_b$ where $c_b\in]0,1]$). 
Moreover, more complex dependencies could be modelled by considering more than one factor and by keeping the same kind of parametrization.
However, the parameter estimation and the likelihood computation would be more complex.
Indeed, the pmf of block $b$ would be defined as a sum of $(d_b+1)^2$ terms, while it is currently a sum of $d_b+1$ terms.

Finally, this model could be an answer to difficult task of the binary data clustering with intra-component dependencies.
Indeed, the clustering aim could be achieved by considering a finite mixture of the proposed distribution.
However, the challenge of model selection would be a complex issue.
Moreover, the model identifiability should be carefully studied.

\appendix 
\section{Proofs of the model properties} \label{app:propositions}

\begin{proof}[Proof of Proposition~\ref{prop:likelihood}]
It suffices to remark that \eqref{pdf} can be decomposed into $d_b+1$ integrals whose bounds are given by the coefficients $\beta_{(b,j)}$. By using the conditional independence between variables given in \eqref{eq:indptcond} and the conditional distribution of $x_j$ given by \eqref{eq:newconditional},  function $p(\bx_{\{b\}}|u_b,\bth_b)$ is a piecewise constant function of $u_b$. Thus, for $u_b\in[\beta_{(b,j)},\beta_{(b,j+1)}[$, $p(\bx_{\{b\}}|u_b,\bth_b)$ is constant and equal to $f_{b}(j)$ defined by \eqref{eq:fbj}. Then,
\begin{align*}
p(\bx_{\{b\}}|u_b,\bth_b) 
&=\int_0^{\beta_{(b,1)}} p(\bx_{\{b\}}|u_b,\bth_b) du
+\sum_{j=1}^{d_b-1} \int_{\beta_{(b,j)}}^{\beta_{(b,j+1)}} p(\bx_{\{b\}}|u_b,\bth_b) du
+\int_{\beta_{(b,d_b)}}^1 p(\bx_{\{b\}}|u_b,\bth_b) du\\
&=\beta_{(b,1)} f_{b}(0;\bth_b) + \sum_{j=1}^{d_b-1} (\beta_{(b,j+1)} - \beta_{(b,j)}) f_{b}(j;\bth_b) + (1-\beta_{(d)}) f_{b}(d_b;\bth_b).
 \end{align*}
\end{proof}

\begin{proof}[Proof of Proposition~\ref{prop:ident}]
We define that the distribution is identifiable if for two vectors of parameters $\bth_{b}=(\alpha_j,\varepsilon_j,\delta_j;j\in\bOmega_b)$ and $\bth_{b}'=(\alpha_j',\varepsilon_j',\delta_j';j\in\bOmega_b)$ such that
\begin{equation}
\forall \bx_{\{b\}}, \; p(\bx_{\{b\}}|\bth_b)=p(\bx_{\{b\}}|\bth_b') \text{ then } \bth_b=\bth_b'. \label{definitionidentifiabilite}
\end{equation}
Without loss of generality, we assume that $\alpha_j\leq \alpha_{j+1}$. The equality $\alpha_j=\alpha_{j'}$ is directly obtained since $\forall j \in \bOmega_b$, $\alpha_j=p(x_j=1|\bth_{b})=p(x_j=1|\bth_{b}')=\alpha_j'$. The proof distinguishes three cases: one variable in the block (\emph{i.e.} $d_b=1$) with the constraints $\delta_{(b,1)}=1$ and $\varepsilon_{(b,1)}=0$; two variables in the block (\emph{i.e.} $d_b=2$) with the constraints $\delta_{(b,1)}=1$ and $\varepsilon_{(b,1)}=\varepsilon_{(b,2)}$; more than two variables in the block (\emph{i.e.} $d_b>2$) with the constraint $\delta_{(b,1)}=1$.
Proofs use the following probability: $\forall (j_1,j_2)\in\bOmega_b,$
\begin{equation}
p(x_{j_1}=1,x_{j_2}=1|\bth_b) =\left\{ \begin{array}{rl}
\alpha_{j_1} \alpha_{j_2} + \varepsilon_{j_1} \varepsilon_{j_2} \alpha_{j_1} (1 - \alpha_{j_2}) & \text{if } \delta_{j_2} = 1\\
\alpha_{j_1} \alpha_{j_2} - \varepsilon_{j_1} \varepsilon_{j_2} \alpha_{j_1} \alpha_{j_2} & \text{if } \delta_{j_2} = 0 \text{ and } \alpha_{j_1} + \alpha_{j_2}<1\\
\alpha_{j_1} \alpha_{j_2} - \varepsilon_{j_1} \varepsilon_{j_2} (1-\alpha_{j_1}) (1 - \alpha_{j_2}) & \text{if } \delta_{j_2} = 0 \text{ and } \alpha_{j_1} + \alpha_{j_2}\geq 1
\end{array}\right. \label{eq:couple}
\end{equation}
If $\delta_{(b,j)}\neq\delta_{(b,j)}'$ then without loss of generality we assume that $\delta_{(b,j)}=1$ and $\delta_{(b,j)}'=0$. From \eqref{eq:couple}, $p(x_{1}=1,x_{j}=1|\bth_b)>\alpha_{(b,1)}\alpha_{(b,j)}=\alpha_{(b,1)}'\alpha_{(b,j)}'>p(x_{1}=1,x_{j}=1|\bth_b')$ but this is in contradiction to \eqref{definitionidentifiabilite}, so $\forall j \in \bOmega_b$, $\delta_{(b,j)}'=\delta_{(b,j)}$. Therefore, we have to prove the equality $\varepsilon_{(b,j)}=\varepsilon_{(b,j)}'$.

\noindent \textbf{Case 1 ($d_b=1$ with $\delta_{(b,1)}=1$ and $\varepsilon_{(b,1)}=0$)}. Then parametrization assumes that only parameter $\alpha$ is free. Equality $\alpha_j=\alpha_j'$ implies $\bth_b=\bth_b'$.

\noindent \textbf{Case 2 ($d_b=2$ with $\delta_{(b,1)}=1$ and $\varepsilon_{(b,1)}=\varepsilon_{(b,2)}$)}. By using constraints $\varepsilon_{(b,1)}=\varepsilon_{(b,2)}$ and $\varepsilon_{(b,1)}'=\varepsilon_{(b,2)}'$ and by using \eqref{eq:couple}, then $\varepsilon_{(b,1)}^2=\varepsilon_{(b,1)}'^2$. Thus, $\bth_b=\bth_b'$.

\noindent \textbf{Case 3 ($d_b>2$ with $\delta_{(b,1)}=1$)}. \eqref{eq:couple} is verified by $\bth$ and $\bth'$. Moreover, we know that $\alpha_j=\alpha_j'$ and $\delta_j=\delta_j'$, for $j=1,\ldots,d$. So, the following system $S$ arises from \eqref{eq:couple} for $(j_1,j_2)=\{(1,2),(1,3),\ldots,(1,d_b),(2,3)\}$
\begin{equation}
(S)=\left\{ \begin{array}{rl}
\varepsilon_{(b,1)} \varepsilon_{(b,2)} &=\varepsilon'_{(b,1)} \varepsilon'_{(b,2)}\\
\varepsilon_{(b,1)} \varepsilon_{(b,3)} &=\varepsilon'_{(b,1)} \varepsilon'_{(b,3)}\\
 \vdots\quad \vdots\quad & = \quad \vdots\quad \vdots\\
\varepsilon_{(b,1)} \varepsilon_{(b,d_b)} &=\varepsilon'_{(b,1)} \varepsilon'_{(b,d_b)}\\
\varepsilon_{(b,2)} \varepsilon_{(b,3)} &=\varepsilon'_{(b,2)} \varepsilon'_{(b,3)}
\end{array}
\right.
\end{equation}
If $\varepsilon_{(b,1)}'\neq\varepsilon_{(b,1)}$ then $\exists t\neq 1$ such that $\varepsilon_{(b,1)}'=t\varepsilon_{(b,1)}$. Then, the first $d_b$ lines of $(S)$ imply that $\forall j=2,\ldots, d_b$, $\varepsilon_{(b,j)}=t\varepsilon_{(b,j)}'$. Since the last line of $(S)$ implies that $\varepsilon_{(b,2)}\varepsilon_{(b,3)}=\varepsilon_{(b,2)}\varepsilon_{(b,3)}/t^2$, positivity of $\varepsilon_{(b,j)}$ permits to conclude that $\varepsilon_{(b,1)}'=\varepsilon_{(b,1)}$, so $\forall j=2,\ldots, d_b$, $\varepsilon_{(b,j)}=\varepsilon_{(b,j)}'$. Thus, $\bth_b=\bth_b'$.
\end{proof}

\begin{proof}[Proof of Proposition~\ref{prop:dependency}]
We denote $p_{hh'}=P(X_{\sigma_b(j)}=h,X_{\sigma_b(j')}|\bomega,\bth)$ with $j<j'$ and $h\in\{0,1\}$ and $h'\in\{0,1\}$. Then
\begin{align*}
p_{11} &= \alpha_{(b,j)}\alpha_{(b,j')} + r \\
p_{01} &= (1-\alpha_{(b,j)})\alpha_{(b,j')} - r\\
p_{10} &= \alpha_{(b,j)}(1-\alpha_{(b,j')}) - r\\
p_{00} &= (1-\alpha_{(b,j)})(1-\alpha_{(b,j')}) + r\\
\end{align*}
where $r=\varepsilon_{(b,j)}\varepsilon_{(b,j')} \beta_{(b,j)}(1-\beta_{(b,j')})$. Thus, \eqref{eq:cramer} is obtained by applying the definition of Cramer's V.
\end{proof}

\section{Details about the M-step of the EM algorithm} \label{app:EM}
By using the definition of $\hat{\alpha}_j$, $\hat{\alpha}_j=n_{10} + n_{01}$ where $n_{10}=\frac{1}{n} \sum_{i=1}^n x_{ij} (1-t_{ij})^{\delta_j}(t_{ij})^{1-\delta_j}$ and $n_{11}=\frac{1}{n} \sum_{i=1}^n x_{ij} (t_{ij})^{\delta_j}(1-t_{ij})^{\delta_j}$. Moreover, the expectation of the complete-data likelihood related to variable $j$ is written as
\begin{align}
L(\hat{\alpha}_j,\delta_j,\varepsilon_j;\tx,\ttt,\bomega) &=
n_{10}\ln((1-\varepsilon_j) (n_{11} + n_{10})) +
n_{11}\ln((1-\varepsilon_j) (n_{11} + n_{10}) + \varepsilon_j) \\
 &+ n_{00}\ln(1-(1-\varepsilon_j) (n_{11} + n_{10})) + 
n_{01}\ln(1-(1-\varepsilon_j) (n_{11} + n_{10})- \varepsilon_j), \nonumber
\end{align}
where  
$n_{00}=\frac{1}{n} \sum_{i=1}^n (1-x_{ij}) (1-t_{ij})^{\delta_j}(t_{ij})^{1-\delta_j}$ and 
$n_{01}=\frac{1}{n} \sum_{i=1}^n (1-x_{ij}) (t_{ij})^{\delta_j}(1-t_{ij})^{\delta_j}$.
For a fixed value of $\delta_j$, the argmax over $\varepsilon_j$ of $L(\hat{\alpha}_j,\delta_j,\varepsilon_j;\tx,\ttt,\bomega)$ is denoted by $\varepsilon_{j|\delta_j}$. The estimation of $\varepsilon_{j|\delta_j}$ is obtained by setting to zero the derivative of $L(\hat{\alpha}_j,\delta_j,\varepsilon_j;\tx,\ttt,\bomega) $ over $\varepsilon_j$. So, remarking that $n_{01}=1-n_{11}-n_{10} - n_{00}$, 
\begin{equation}
\frac{n_{11} + n_{00} -1}{1-\varepsilon_{j|\delta_j}} + \frac{n_{11} (1-n_{11}-n_{10})}{n_{11} + n_{10} + \varepsilon_{j|\delta_j}(1-n_{11} - n_{10})} + 
\frac{n_{01} (n_{11}+n_{10})}{(n_{11} + n_{10}) \varepsilon_{j|\delta_j} + (1-n_{11} - n_{10})} = 0.
\end{equation}
This equation is equivalent to the following quadratic equation
\begin{equation}
\varepsilon^2 A + \varepsilon B + C=0, \label{eq:quadra}
\end{equation}
where $A=-(n_{11} + n_{10})(1-n_{11}-n_{10})$,
$B=n_{11}(n_{11} + n_{10}) + n_{00}(1-n_{11} - n_{10}) - (n_{11} + n_{10})^2 - (1 - n_{11} - n_{10})^2$ and where
$C=n_{11}(1-n_{11} - n_{10})  + n_{00}(n_{11} + n_{10}) + A$. Let $s_1$ and $s_2$ be the two solutions of \eqref{eq:quadra}:
\begin{equation}
s_1=\frac{-B-\sqrt{\Delta}}{2A} \text{ and }s_2=\frac{-B+\sqrt{\Delta}}{2A} ,
\end{equation}
where $\Delta=B^2-4AC$.  By noting that $\varepsilon_j\in]0,1[$, and that $s_1=\frac{(n_{11} + n_{10}) n_{10} + (1 - n_{11} - n_{10})n_{01} }{-2(n_{11} + n_{10})(1 - n_{11} - n_{10})}<0$, we conclude that $\varepsilon_{j|\delta_j}=\max(0,s_2)$.

\section*{Consistency of the HAC-based procedure} \label{app:consistency}
The proof of Proposition~\ref{prop:consistency} is done in three steps. First, we show that the HAC-based procedure applied to the theoretical Cramer's matrix is consistent. Second, we show that this result holds in a neighbourhood of the theoretical Cramer's matrix. Third, we conclude by using the convergence in probability of the empiric Cramer's matrix to the theoretical one.

Let $M^0\in[0,1]^{d\times d}$ be the dissimilarity matrix based on Cramer's V computed with the true distribution defined by model $\bomega^0$ and its parameters $\bth^0$. So, for $1\leq j,j' \leq d$
\begin{equation}
M^0(j,j')= 1 - V^0(X_{j},X_{j'})
\end{equation}
with $V^0(X_{j},X_{j'})$ is the theoretical Cramer's V between $X_{j}$ and $X_{j'}$ defined by
\begin{equation}
 V^0(X_{j},X_{j'})=\sqrt{\sum_{h=0}^1\sum_{h'=0}^1 \frac{\left(P(X_j=h,X_{j'}=h'; \bomega^0, \bth^0) - P(X_j=h; \bomega^0, \bth^0)P(X_{j'}=h'; \bomega^0, \bth^0)\right)^2}{P(X_j=h; \bomega^0, \bth^0)P(X_{j'}=h'; \bomega^0, \bth^0)}},
\end{equation}
Since the true model $\bomega^0$ involves independence between blocks of variables,  for $1\leq j,j' \leq d$ with $\omega_{j}^0\neq\omega_{j'}^0$, $M^0(j,j')=1$. We denote by $\mu^0$ the greatest value of $M^0$ when the variables belong to the same block for the true model $\bomega^0$
\begin{equation}
\mu^0=\arg\max_{\{(j,j'): \omega_{j}^0 = \omega_{j'}^0\}}M^0(j,j').
\end{equation}
Note that $\mu^0<1$ since the variables affiliated into the same block are dependent.
Finally, $\bOmega^{[r]}=(\bOmega_b^{[r]};b=1,\ldots,d)$ denotes the partition provided by the HAC at its iteration $[r]$, where $\Omega_b^{[r]}$ is the set of the indices of the variables affiliated to block $b$ at iteration $[r]$. We consider that the HAC is used with a classical dissimilarity measure $D(.,.)$ (min, max, mean or Ward).

\begin{proposition}\label{prop:recurrence}
 If $\exists (j_1,j_2)$ $\in$ $\{1,\ldots,d\}^2$ with $\omega_{j_1}^0=\omega_{j_2}^0$ and with $j_1\in\bOmega_{b_1}^{[r]}$, $j_2\in\bOmega_{b_2}^{[r]}$  and $b_1 \neq b_2$, then
\begin{equation}
\forall b, \forall (j,j') \in \bOmega_b^{[r+1]}: \omega_j^0=\omega_{j'}^0. \label{eq:recurence}
\end{equation}
\end{proposition}
\begin{proof}
At iteration $[0]$, each variable is affiliated into its own block, so $\bOmega_b^{[r]}=\{b\}$ for $b\in\{1,\ldots,d\}$. Let $(j_1^{[0]},j_2^{[0]})=\arg\min_{(j_1,j_2)} M^0(j_1,j_2)$, then
\begin{equation}
\bOmega_b^{[1]}=\left\{ \begin{array}{rl}
\bOmega_b^{[0]} & \text{if } b\neq j_1^{[0]} \text{ and } b\neq j_2^{[0]} \\
\bOmega_{j_1^{[0]}}^{[0]} \cup \bOmega_{j_2^{[0]}}^{[0]} & \text{if } b=j_1^{[0]}\\
\emptyset & \text{if } b=j_2^{[0]}
\end{array} \right. .
\end{equation}
The $\bOmega^{[1]}$ verifies \eqref{eq:recurence}.

At iteration $[r]$, by definition $\forall b$, $\forall (j,j')\in\bOmega_b^{[r]}$: $\omega_{j}^0=\omega_{j'}^0$. Let the couple  $(b_1^{[r]},b_2^{[r]})=\arg\min_{(b_1,b_2) \text{with} b_1\neq b_2} D(\bOmega_{b_1}^{[r]},\bOmega_{b_2}^{[r]})$. There are two cases to be considered, for all $j_1\in\bOmega_{b_1}^{[r]}$ and  $j_2\in\bOmega_{b_2}^{[r]}$,
\begin{itemize}
\item if $\omega_{j_1}^0\neq \omega_{j_2}^0$ then $D(\bOmega_{b_1}^{[r]},\bOmega_{b_2}^{[r]}) = 1$.
\item if $\omega_{j_1}^0= \omega_{j_2}^0$ then $D(\bOmega_{b_1}^{[r]},\bOmega_{b_2}^{[r]})\leq \mu^0$.
\end{itemize}
Since $\mu^0<1$, \eqref{eq:recurence} is verified.
\end{proof}

\begin{corollary}[Consistency with theoretical matrix]
The HAC based on the dissimilarity matrix $M$ provides the true model at its iteration $d-\B^0$ where $\B^0$ is the number of blocks defined by $\bomega^0$.
\end{corollary}

\begin{proof}
It is the only partition of $\B^{[0]}$ classes which respects Proposition~\ref{prop:recurrence}.
\end{proof}

\begin{corollary}[Consistency in a neighbourhood of the theoretical matrix]  \label{cor:neighbourhood}
The HAC based on dissimilarity matrix$M$ belonging to a neighbourhood of $M^0$, denoted by $N(M^0)$, provides the true model at its iteration $d-\B^0$ where
\begin{equation}
N(M^0)=\left\{ M\in[0,1]^{d\times d} \text{ with } |M(j,j') - M^0(j,j')|<\frac{1-\mu^0}{2}\right\}.
\end{equation}
\end{corollary}

\begin{proof}
Proof is based on the same reasoning as the proof of Proposition~\ref{prop:recurrence}, since we have
\begin{align*}
\left\{ \begin{array}{rl}
M(j,j')>\frac{1+\mu^0}{2} & \text{ if } \omega_j^0 \neq \omega_{j'} \\
M(j,j')<\frac{1+\mu^0}{2} & \text{ if } \omega_j^0 = \omega_{j'}
\end{array} \right. .
\end{align*}
\end{proof}

\begin{proof}[Proof of Proposition~\ref{prop:consistency}]
The Law of Large numbers implies that the observed probability of each couple $(j,j')$ converges in probability to its true value: for any $h\in\{0,1\}$ and $h'\in\{0,1\}$
\begin{equation}
\hat{p}_{hh'} \stackrel{pr}{\rightarrow} P(X_j=1,X_{j'}=1; \bomega^0, \bth^0),
\end{equation}
where $p_{hh'}=\frac{1}{n} \sum_{i=1}^n \mathds{1}_{x_{ij}=h}\mathds{1}_{x_{ij'}=h'}$.

The empirical Cramer's V  denoted by $\hat{V}$ is a continuous function of $\hat{p}_{hh'}$, since it is defined by
\begin{equation}
\hat{V}(X_{j},X_{j'})=\sqrt{\sum_{h=0}^1\sum_{h'=0}^1 \frac{\left(\hat{p}_{hh'}- \hat{p}_{h\bullet}\hat{p}_{\bullet h'}\right)^2}{\hat{p}_{h\bullet}\hat{p}_{\bullet h'}}},
\end{equation}
where $\hat{p}_{h\bullet}=\hat{p}_{h0}+\hat{p}_{h1} $ and $\hat{p}_{\bullet h'}=\hat{p}_{0 h'}+\hat{p}_{1 h'}$.
 Thus, the Mapping theorem (see for instance Theorem 2.7 page 21 of \citet{billingsley2013convergence}) implies that $\hat{V}$ converges in probability to $V^0$. So,
 \begin{equation}
 P(M \in N(M^0)) \stackrel{n\to\infty}{\rightarrow} 1.
\end{equation}  
Thus, by applying Corollary~\ref{cor:neighbourhood}, the probability that $\bomega^0$ belongs to the model subset provided by the HAC procedure is equal to one. The consistency of the BIC criterion concludes the proof.
\end{proof}

\bibliographystyle{apalike}
\bibliography{biblio}
\end{document}